\documentclass[twocolumn,prb,showpacs,preprintnumbers,amsmath,amssymb,superscriptaddress,reprint]{revtex4-1}[11pt]
\usepackage{graphicx}
\usepackage{dcolumn}
\usepackage{bm}
\usepackage{graphicx}
\usepackage{physics}
\usepackage{epstopdf}
\usepackage[utf8]{inputenc}
\usepackage[english]{babel}
\usepackage{verbatim}
\usepackage{amsmath,amsthm}
\usepackage{amsfonts}
\usepackage{mathtools}
\usepackage[usenames, dvipsnames]{color}

\usepackage{hyperref}
\begin{document}

\title{A Bayesian view of Single-Qubit Clocks, and an Energy versus Accuracy tradeoff}
\author{Manoj Gopalkrishnan}
\email{manoj.gopalkrishnan@gmail.com}
\affiliation{Department of Electrical Engineering, Indian Institute of Technology Bombay, Powai, Mumbai-400076, India}

 \author{Varshith Kandula}
 \affiliation{Department of Electrical Engineering, Indian Institute of Technology Bombay, Powai, Mumbai-400076, India
}

\author{Praveen Sriram}
 \affiliation{Department of Electrical Engineering, Indian Institute of Technology Bombay, Powai, Mumbai-400076, India
}

\author{Abhishek Deshpande}
\affiliation{School of Technology and Computer Science,
Tata Institute of Fundamental Research, Mumbai-400005, India\\
Department of Mathematics,
Imperial College London, London SW7 2AZ , UK}

\author{Bhaskaran Muralidharan}
\email{bm@ee.iitb.ac.in}
 \affiliation{Department of Electrical Engineering,\\
Indian Institute of Technology Bombay, Powai, Mumbai-400076, India
}

\date{\today}

\begin{abstract}
We bring a Bayesian approach to the analysis of clocks. Using exponential distributions as priors for clocks, we analyze how well one can keep time with a single qubit freely precessing under a magnetic field. We find that, at least with a single qubit, quantum mechanics does not allow exact timekeeping, in contrast to classical mechanics which does. We find the design of the single-qubit clock that leads to maximum accuracy. Further, we find an energy versus accuracy tradeoff --- the energy cost is at least $k_BT$ times the improvement in accuracy as measured by the entropy reduction in going from the prior distribution to the posterior distribution. We propose a physical realization of the single qubit clock using charge transport across a capacitively-coupled quantum dot.

\end{abstract}
\maketitle

\renewcommand\comment[1]{{\color{magenta} $\star$#1$\star$}}
\newcommand\e{\operatorname{e}}
\theoremstyle{plain}
\newtheorem{theorem}{Theorem}[section]
\newtheorem{lemma}{Lemma}[section]

\newenvironment{bullets}%
{\begin{list}
		{\noindent\makebox[0mm][r]{$\bullet$}}
		{\leftmargin=5.5ex \usecounter{enumi}
			 		
			\topsep=1.5mm \itemsep=-.75ex}
	}
	{\end{list}}

\section{\label{sec:intro}INTRODUCTION}
A clock is a device that couples a periodic or approximately periodic motion to a counter that increments upon ``ticks'' of the periodic motion. Classical mechanics allows harnessing periodic motion from a simple harmonic oscillator \cite{Allan} to build perfectly accurate clocks, at least in principle and in the absence of noise. Do the laws of quantum mechanics allow clocks with perfect inter-tick durations? One difficulty manifests immediately. Though a quantum system may display periodic motion, quantum measurement only provides partial information about the full quantum state. The first question we address here is: what are the limits to accuracy of inter-tick durations for resource-limited quantum systems?

In classical mechanics, in the absence of noise, clocks need not dissipate any energy. The rotation of the earth may be set forward as an example that comes very close to this ideal. In practice, man-made clocks require energy: wall clocks run on batteries, mechanical pendulum clocks and watches run down and need to be wound up. Do the laws of quantum mechanics require clocks to be dissipating? This is the second question we address.

We make a step towards addressing these questions by describing clocks as information processing devices that employ Bayesian inference, and use this framework to analyze the case of a clock constructed from a single qubit.

\textbf{Contributions:}
\begin{itemize}

\item Our approach in Section~\ref{sec:clocks} brings to the fore the role of information processing in the keeping of time. In Subection~\ref{subsec:bayes}, we connect the problem of timekeeping to Bayesian inference. In Subection~\ref{subsec:rv}, we describe the time between ticks in the language of random variables, and argue for treating exponential random variables as free resources, and hence as reasonable Bayesian priors.%we find that it plays in timekeeping the same role as in thermodynamics, where in the form of the Gibbs distribution it models thermal equilibrium.

\item The minimal example of periodic motion in quantum mechanics is a precessing spin modeled by a single qubit. We show in Section~\ref{sec:qubitclock} how to construct the most accurate clock possible given the resource constraint of a single precessing spin and a process that generates events with exponential inter-arrival times. Our results show that within these resource constraints quantum mechanics does not allow perfectly accurate timekeeping.

\item We show in Section~\ref{sec:tradeoff} that there is an energy versus accuracy tradeoff for keeping time with a single qubit. The smaller the desired spread of uncertainty around the time of a tick, the greater the amount of energy required. Specifically we prove in Theorem~\ref{thm:energyaccuracy} that the amount of energy required is at least $k_B T$ times the accuracy gain as measured by reduction in entropy of the inter-tick distribution. 

\item Our results encourage us to speculate on two new principles for quantum timekeeping. First, our results of Section~\ref{sec:qubitclock} lead us to speculate that resource-constrained quantum systems may not allow perfect timekeeping. Second, Theorem~\ref{thm:energyaccuracy} leads us to speculate that there may be an energy versus accuracy tradeoff for timekeeping which manifests in a form reminiscent of the Szilard-Landauer principle (Section~\ref{sec:tradeoff}), except that the relevant entropy is defined on the time variable.

\item In Section~\ref{sec:setup}, we suggest a physical implementation of our proposal via a charge transport set up involving two capacitively coupled quantum dots in an attempt to outline a scheme for estimating tunneling times.
\end{itemize}

\section{An Information Processing View of Clocks}\label{sec:clocks}

\subsection{Bayesian Inference}\label{subsec:bayes}
There is an apparent paradox at the heart of timekeeping. Two readings of a clock face inform us of the duration of time only upto a periodic factor, yet we are never confused about the actual time elapsed. Consider this thought experiment. We make two observations on a typical wall clock with markings from $1$ to $12$ that is assumed to be functioning correctly. The first observation reports the hour hand on $7$ and the minute hand on $0$, and the second observation (one hour later) reports the hour hand on $8$ and the minute hand on $0$. According to this clock, going strictly by the observation, the time elapsed equals $12 n + 1$ hours where $n$ is a natural number. The natural number $n$ has to be determined using means external to the clock. In practice, we are not often confused about the value of $n$, and can confidently assert that $1$ hour has elapsed between the two observations. Why is this so? Where did we get the side information that allows us to confidently assert that $n=0$? 

The answer is that we have an ``a priori'' sense of the passage of time, which comes from observing various events or ``ticks'' that are constantly happening around us. Observation of the clock face allows us to refine this prior and infer how much time has elapsed. In this view of timekeeping, clocks refine our ``a priori'' notion of time. A clock is then fully described by specifying the prior notion of time, as well as a new observation that allows us to update this prior notion of time, and a counter to accumulate successive estimates of time elapsed. We thus sidestep the question ``what is a clock,'' and focus on the question of improving a given clock with the help of side information.

The Bayesian approach to modelling uncertainty is to introduce probability distributions. We will describe our prior notion of time by means of a random variable $T$ taking values in the positive reals, for example representing the time that elapses between consecutive observations of a clock. Just before we make a new observation, we are uncertain about exactly what the time is, with the uncertainty described by the spread of the distribution of $T$. For example, in our thought experiment, $T$ represents our uncertainty about the time at the second instant when we decided to look at the clock, just before we noted the hour hand on $8$ and the minute hand on $0$. We don't know that exactly one hour has elapsed, but it is likely that most of the probability is concentrated around the one hour mark, and very little probability is around $12n+1$ hours for larger values of $n$. This is the side information we are using to decide that $n=0$ with high probability.

The physical experiment that we perform to refine our notion of time --- for example, reading the face of the clock --- gives us a finite number of outcomes. Let $S$ be a random variable taking values in a finite set. Suppose we get to observe $S$, and find that the event $S = s$ is true. We have obtained some information about the random variable $T$ from the correlation between the random variables $T$ and $S$. The random variable $T_s := T|(S=s)$ is obtained from $T$ by conditioning on this information. For an interval $I\subseteq\mathbb{R}_{\geq 0}$, the posterior probability ${\Pr[T_s\in I] = \Pr[T\in I\mid S=s]}$ of $T_s$ is computed by Bayes' law:
\begin{equation}\label{eq:bayes}
	\Pr[T\in I\mid S=s] = \frac{\Pr[S=s \mid T\in I] \,\Pr[T\in I]}{\Pr[S=s]}
\end{equation}
\subsection{Clocks as Random Variables}\label{subsec:rv}

Consider a random variable $T$ that takes values in $\mathbb{R}_{\geq 0}$ and has expected value $E[T] = 1/\lambda$. The best such random variable for accuracy of timekeeping is a \textbf{delta distribution} $\delta_{1/\lambda}$, because this corresponds to complete certainty. The worst such random variable for accuracy of timekeeping is one whose distribution is as spread out as possible. If use differential entropy to measure the amount that the probability density $f(t)=\operatorname{Pr}[T\in(t,t+dt)]$ is spread out, we need to find the random variable $T^*$ that maximizes the \textbf{differential entropy} $h[T] = -\int_{t=1}^\infty f(t) \log f(t) dt$ subject to the constraint that $E[T] = 1/\lambda$.  It is well-known that the unique solution to this maximum entropy problem is the \textbf{exponential distribution} $T^*$ which obeys $\Pr[T^* > t] = e^{-\lambda t}$ and has probability density $\Pr[T\in (t, t+dt)] = \lambda e^{-\lambda t} dt$. 

In our resource-theoretic treatment of clocks, we will treat exponential random variables as free resources, since they correspond to the weakest assumption we can make on our prior sense of time. This is reminiscent of the heat bath which is a free resource in thermodynamics, and is modeled by an exponential distribution (the Gibbs distribution), and is in the spirit of the MaxEnt philosophy~\cite{jaynes1957information}.

Apart from the differential entropy, we will find it useful to introduce another metric to report on the spread of a probability distribution. For random variables $T$ taking values in the positive reals, we define the \textbf{quality factor} $Q[T]$ as
\[
				Q[T] = \frac {E[T]}{\sqrt{E[T^2] - (E[T])^2}} 
\]
A higher quality factor would imply a narrower distribution and thus a higher probability for the outcome of the random variable to be close to the mean. The quality factor is a dimensionless quantity. In particular, it is invariant to change of the units by which we measure time. For exponential random variables $T$ the quality factor is $Q[T]=1$. If $T_1, T_2, \dots, T_n$ are $n$ independent, identically-distributed exponential random variables, then their sum $T = T_1 + T_2 + \dots T_n$ has quality factor $\sqrt{n}$. 

Thus one way to obtain accurate timekeeping is by keeping count of events with independent and identically distributed inter-arrival times, and declaring $n$ events to be one tick. This is the idea behind water clocks. Though the duration for each single drop to fall is highly random, the duration for the entire vessel to be emptied has a much higher quality factor. Another way is to couple the random variable $T$ to some periodic motion, which we explore in the next section.

\section{The single qubit clock}\label{sec:qubitclock}
Given a prior sense of time, we want to couple it with some physical experiment that will refine our estimate of time. It is natural to consider an experiment corresponding to a physical system that is undergoing periodic motion, so that we can exploit the periodicity to get an accurate time estimate.

A minimal example of periodic motion in quantum mechanics is a spin freely precessing around an axis, described by a single qubit evolving with respect to a time-invariant Hamiltonian. Another motivation for considering a single qubit system is the hope that general quantum clock systems can be described in terms of multiple qubits, so the current analysis may serve as a building block.

Given an arbitrary time-invariant Hamiltonian acting on a single qubit, let us call its ground state as $|0\rangle$ and its other eigenstate as $|1\rangle$, so that without loss of generality,
\begin{align}\label{eqn:ham}
\hat{H} = \frac{\hbar \omega}{2} |1\rangle\langle 1|-\frac{\hbar\omega}{2} |0 \rangle \langle 0|,
\end{align}
where $\hbar$ is the reduced Planck's constant. The qubit's unitary evolution can be visualized along the Bloch 
sphere~\cite{slichter}(Figure~\ref{fig:sub2}).  The point $(\theta, \phi)$ corresponds to the state $| \psi \rangle = \cos (\theta/2)| 0 \rangle + e^{i\phi}\sin(\theta/2)$. Through evolution, the angle $\theta$ remains constant, so that the circles of latitude are invariants of motion. The Bloch sphere ``spins'' anticlockwise around the $z$-axis with an angular velocity of $\omega$.

Imagine that an event whose inter-arrival times are exponentially distributed is perfectly coupled to a projective measurement of the qubit. For example, whenever a radioactive decay occurs, the spin gets measured. The coupling between the spin measurement and the event may be achieved via electrostatic coupling, exchange interaction or any other mechanism dictated by the coupling Hamiltonian, the details of which need to be carefully considered while designing a physical apparatus. Assume that these interactions are ``instantaneous'' and ideal, resulting in a simultaneous measurement of the qubit state when the event triggers. 

Our first task is to infer the time as best we can, from the observation of the measurement outcome. Our next task is to figure out how to tune the angular velocity $\omega$, the arrival rate $\lambda$ of the exponential process, the initial position of the qubit, and the measurement axis for the projective measurement to get the best clock possible with a single qubit coupled to an exponential random variable via a projective measurement. In Subsection~\ref{subsec:q}, we will analyze a special case where the precessing spin is on the equator of the Bloch sphere, and the projective measurement states are also on the equator of the Bloch sphere. In Subsection~\ref{subsec:q1}, we will argue that our solution to this special case is, in fact, the optimal clock possible.

\subsection{Equatorially-precessing qubit}\label{subsec:q}

Suppose the precessing spin starts on the equator of the Bloch sphere (Figure~\ref{fig:equator}) so the state is $|\psi(0)\rangle =\frac{1}{\sqrt2}(|0\rangle+e^{i\phi}|1\rangle)$. Then the qubit's state at time $t$ is given by:
\begin{equation}
\begin{split}
|\psi(t)\rangle & = e^{-i\hat{H}t/\hbar}|\psi(0)\rangle \\
&= \frac{1}{\sqrt2}(e^{-i\omega t/2}|0\rangle+e^{i(\omega t/2+\phi)}|1\rangle)
\end{split}
\end{equation}
\begin{figure}
	\includegraphics[width=0.28\textwidth, height=0.23\textwidth]{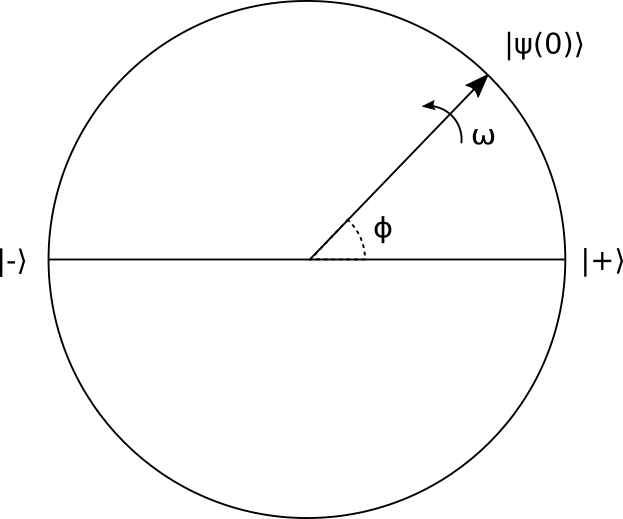}
	\quad
	\caption{A spin precesses with angular velocity $\omega$ on the equator of the Bloch sphere. Arrival of the exponential random variable with rate $\lambda$ triggers a projective measurement in the $\{|+\rangle,|-\rangle\}$ basis. The time reported depends on the measurement outcome.}
	\label{fig:equator}
\end{figure}

Suppose after an unknown passage of time $t$, the qubit were measured  by orthogonal projection to states $|+\rangle = \frac{1}{\sqrt2}(|0\rangle+|1\rangle)$ and $|-\rangle = \frac{1}{\sqrt2}(|0\rangle-|1\rangle)$.
Then according to the Born rule, the outcome of this measurement will be a ``spin'' random variable $S(\omega t)$ taking values ``+'' and ``-'' with 
\begin{align*}
\Pr[S(\omega t)  = +] &= \cos^2\frac{\phi + \omega t}{2}\
\\\Pr[S(\omega t)=-] &= \sin^2\frac{\phi + \omega t}{2}.\
\end{align*} 
Our task is to infer the time $t$ from the observation of $S(\omega t)$.

Let $T$ be an exponential random variable with rate $\lambda$. If $T$ triggers the measurement of the spin then we want to consider the distribution of the random variable $S(\omega T)$. 

\begin{equation*}
\Pr[S(\omega T) = +] = \int_{t=0}^{\infty} \lambda e^{-\lambda t}\Pr[S(\omega t)  = +]dt\
\end{equation*}
\begin{equation}
= \left[\frac{1}{2} + \frac{ 
\cos \phi-(\omega/\lambda)\sin \phi}{2(1+(\omega/\lambda)^2)}\right]
\end{equation}

\begin{equation}
\Pr[S(\omega T) = -] = \left[\frac{1}{2} -\frac{ 
\cos \phi-(\omega/\lambda)\sin \phi}{2(1+(\omega/\lambda)^2)}\right]
\end{equation}

Let $T_+=[T\mid S(\omega T)=+]$ be the posterior random variable upon measurement of spin `+',  and let $T_{-}=[T\mid S(\omega T)=-]$ be the posterior random variable upon measurement of spin `-'. 
\begin{figure}[htb!]
	\includegraphics[width=0.45\textwidth, height=0.3\textwidth]{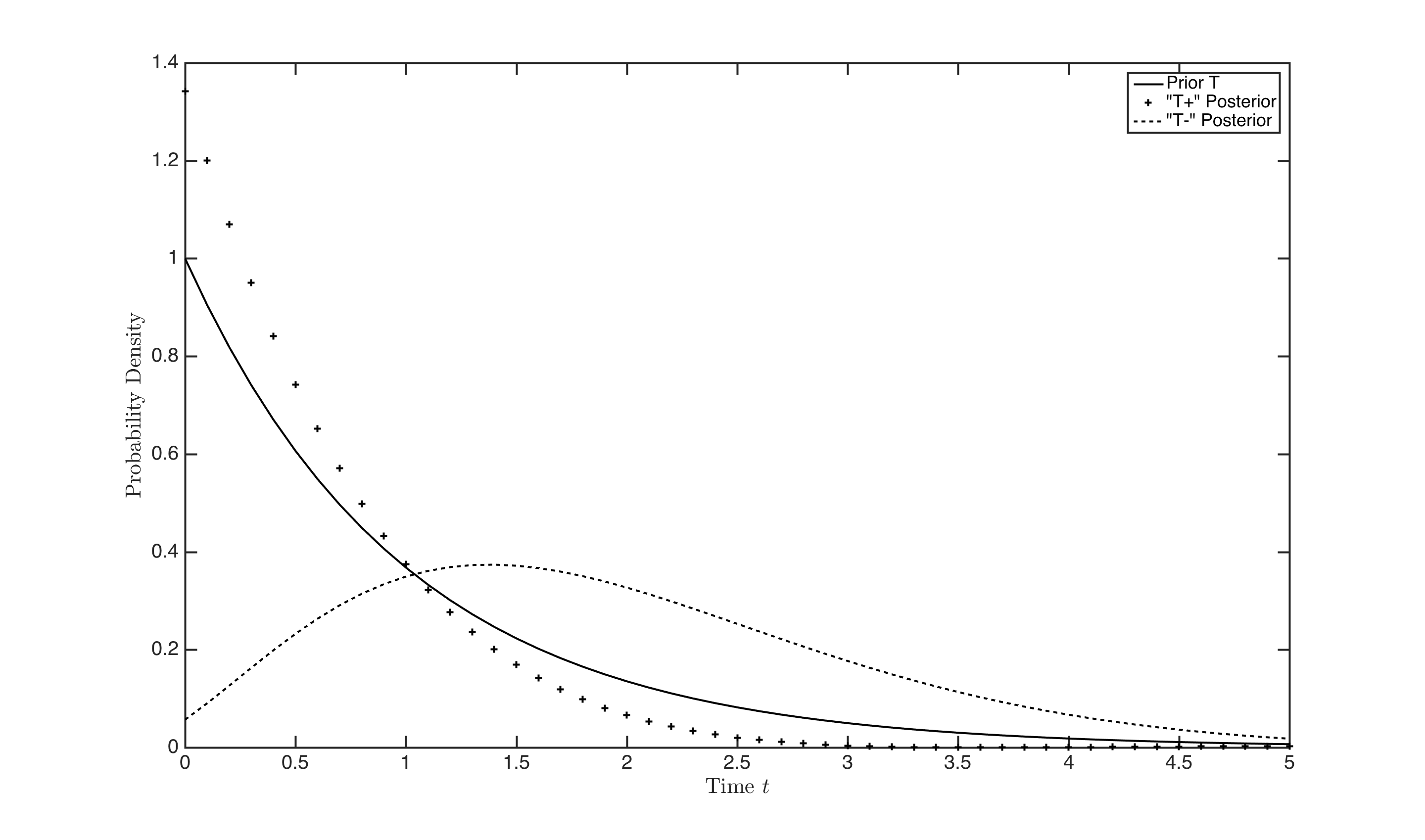}
	\quad
	\caption{Probability density functions for the prior $T$ and the two posteriors $T_{+}$ and $T_{-}$ at $\omega/\lambda =  0.80654$ and $\phi=0.246576$. The declared time for a measurement `-' is the mean of $T_-$, which is significantly different from the mean of $T$.\label{fig:posterior}}
	
\end{figure}
Immediately after the measurement, the state of the qubit has collapsed to either $|+\rangle$ or $|-\rangle$. The qubit evolves under the unitary dynamics of its Hamiltonian until the next event occurs with an exponential waiting time. Upon the occurrence of this event, the qubit is again measured, but now with respect to a new ``rotated'' basis $\{|+'\rangle,|-'\rangle\}$ which is at an angle $\phi$ clockwise to the previous measurement basis. This is how successive ticks are obtained. After $n$ measurements, if `+' was the measurement outcome a total of $n_+$ times, and `-' was the measurement outcome a total of $n-n_+ = n_-$ times, then we declare the time to be $n_+ E[T_{+}] + n_- E[T_{-}]$.

We define the  \textbf{expected quality factor} as
\begin{multline}\label{eq:avgq}
Q[T \mid S(\omega T)] = \Pr\left[S(\omega T) = +\right]\,Q[T_{+}]\,\, \\
+ \,\,\Pr\left[S(\omega T) = -\right]\,Q[T_{-}]
\end{multline}
Note that $Q[T\mid S(\omega T)]$ is periodic in $\phi$ with a fundamental period of $\pi$. This is because changing $\phi$ by $\pi$ is equivalent to interchanging $|+\rangle$ with $|-\rangle$. But since the choice of calling one of the basis vectors as $|+\rangle$ and the other $|-\rangle$ was purely arbitrary, they can't affect any aspect of physical reality, and hence all the metrics that we extract from this experiment would essentially remain the same.

By Bayesian Inference~\eqref{eq:bayes}, the densities of $T_{+}$ and $T_{-}$ are (Figure~\ref{fig:posterior}):
\begin{align*}
\Pr[T_{+} \in (t, t+dt)] &= \frac{{2(1+(\omega/\lambda)^2)}\lambda e^{-\lambda t}\cos^2\frac{\phi + \omega t}{2}dt}{1+(\omega/\lambda)^2 + 
\cos \phi-(\omega/\lambda)\sin \phi} \\
\Pr[T_{-} \in (t, t+dt)] &= \frac{{2(1+(\omega/\lambda)^2)}\lambda e^{-\lambda t}\sin^2\frac{\phi + \omega t}{2}dt}{1+(\omega/\lambda)^2 - 
\cos \phi+(\omega/\lambda)\sin \phi}
\end{align*}

The expected quality factor $Q[T \mid S(\omega T)]$ is a function of the ratio $\omega/\lambda$ and $\phi$. Figure~\ref{fig:q} shows an intensity plot of $Q[T \mid S(\omega T)]$ as a function of $\omega/\lambda$ and $\phi$. The quality factor attains a maximum value of $1.2184$ at $\omega/\lambda = 0.80654$ and $\phi=0.246576$. This is an improvement over the quality factor $Q[T]=1$ for the exponential random variable $T$.
\begin{figure}			
			\includegraphics[width=0.45\textwidth, height=0.32\textwidth]{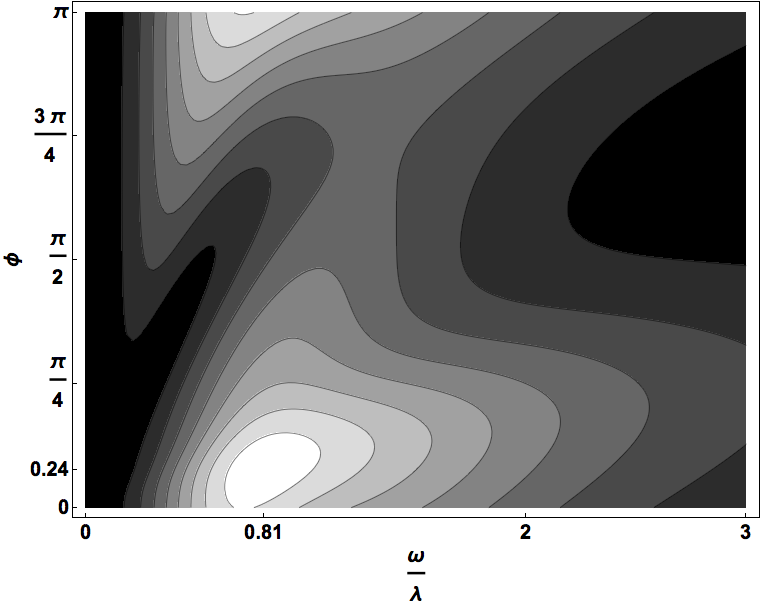}
			\quad
			\caption{Intensity plot of expected quality factor $Q[T \mid S(\omega T)]$ shows maximum at $\omega/\lambda=0.81,\phi=0.24$. \label{fig:q}}
			
\end{figure}
In the next subsection, we argue that this is the best quality factor attainable, even when a more general initial state and measurement basis are considered.

\subsection{\label{subsec:q1} General Single Qubit Clock}
\begin{figure}
	\includegraphics[width=0.37\textwidth, height=0.3\textwidth]{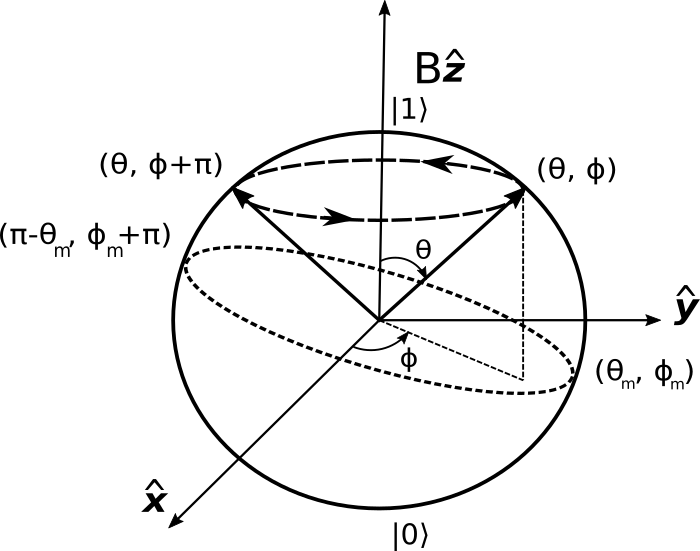}
	\quad
	\caption{\label{fig:sub2}A spin precesses on a circle of latitude making angle $\theta$ with the north pole. Arrival of the rate-$\lambda$ exponential random variable triggers a projective measurement in direction $(\theta_m,\phi_m)$. We maximize quality factor against $\phi,\phi_m,\theta,\theta_m, \omega,\lambda$ and find $\theta=\theta_m=\pi/2$.}	
\end{figure}

The qubit starts in a general initial state $| \theta_0, \phi_0 \rangle$, which in the computational basis is $ \cos(\theta_0/2)|0\rangle+\sin(\theta_0/2)e^{i\phi_0}|1\rangle$. Through time it traces a circle of latitude on the Bloch sphere (Figure~\ref{fig:sub2}). The state at time $t$ is given by 
\begin{equation*}
|\psi(t)\rangle  = \cos(\theta_0/2)e^{-i\omega t/2}|0\rangle+\sin(\theta_0/2)e^{i(\omega t/2+\phi_0)}|1\rangle
\end{equation*}
We denote the measurement basis by the antipodal points on the Bloch Sphere $|\theta_m, \phi_m\rangle$ and $|\pi - \theta_m, \pi+\phi_m\rangle$. The probabilities now are, 
\begin{equation}
\begin{split}
\Pr[S(\omega t)  = +] & =\big|   \cos(\theta_m/2) \cos(\theta_0/2) +\\
 & \quad \sin(\theta_m/2)\sin(\theta_0/2)e^{-i(\phi+\omega t)} \big|^2
\end{split}
\end{equation}
\begin{equation*}
\Pr[S(\omega T) = +] = \int_{t=0}^{\infty} \lambda e^{-\lambda t}Pr[S(\omega t)  = +]dt\
\end{equation*}
\begin{equation}\label{eqn:Sdist}
\begin{split}
= \left[\dfrac {\splitdfrac{(1+(\omega/\lambda)^2)(1+\cos\left(\theta_0\right)\cos\left(\theta_m\right))+}{
(\sin\left(\theta_0\right)\sin\left(\theta_m\right))(\cos\left(\phi\right)+(\omega/\lambda)\sin\left(\phi\right))}}{2(1+(\omega/\lambda)^2)}\right]
\end{split}
\end{equation}
with $\phi=\phi_0-\phi_m$. Maximizing $Q[T \mid S(\omega T)]$ over $\omega/\lambda, \theta_0, \theta_m$, and $\phi$, we find that $\theta_0=\theta_m=\pi/2$. (See Appendix~\ref{sec:app} for details.) Thus the maximum quality factor attainable in the most general case with a single qubit can already be obtained with the system analyzed in the previous section. 

Since Bayesian inference makes optimum use of the information available from coupling the random variables $T$ and $S(\omega T)$, we conclude that with these resource constraints, no further improvement is possible. In particular, with these resource constraints, quantum mechanics disallows perfectly accurate timekeeping.

\section{Energy-Accuracy Tradeoff}\label{sec:tradeoff}
Does it require energy to keep time? Specifically, must it require more energy to keep time more accurately? We show in this section that the answer is yes for our system. Further the excess energy required is lower bounded by $k_B T$ times the improvement in accuracy, where $k_B$ is Boltzmann's constant and $T$ is temperature. We first describe our metrics for accuracy and energy.

In this section, we will describe the accuracy of an inter-tick duration by its differential entropy. Thus improvement in accuracy is measured by the decrease in differential entropy. More precisely, if $T$ is an exponential random variable of mean $1/\lambda$, a straightforward calculation shows $h[T] = 1-\log\lambda$. After the spin random variable $S(\omega T)$ is observed, the conditional differential entropy $h[T\mid S(\omega T)]$ is, by definition,
${\Pr[S(\omega T) = +]\,h[T_+]\,\,+ \,\,\Pr[S(\omega T) = -]\,h[T_-]}$. The increase in accuracy is measured by the decrease in entropy $h[T] - h[T\mid S(\omega T)]$ caused by observing the coupled spin.

For energy accounting, we focus on the energy required to measure the spin. Let $p = \Pr[S(\omega T)=+]$. Then the spin random variable has an entropy $H[S(\omega T)] = - p \log p - (1-p)\log (1-p)$. By the Szilard-Landauer principle~\cite{szilard,landauer,Manoj}, we declare $k_B T H[S(\omega T)]$ as the energy cost for the spin measurement. The dissipation of this energy happens when the spin collapses from its pure state to the mixed state described by $p\,{|+\rangle\langle+|}\,\,+\,\,(1-p)\,|-\rangle\langle -|$. Work is done on the system when learning the outcome of the measurement, which takes us from the mixed state to the pure state $|+\rangle$ or $|-\rangle$ as reported by the measuring device. Learning the outcome of the measured spin corresponds to an ``erasure'' since the entropy of the qubit must decrease from $H[S(\omega T)]$ to $0$.

There may be other energy costs to the device apart from the measurement of the spin. Here we ignore other costs, so that our metric forms a lower bound on the true energy requirement. The next theorem states that \textbf{the energy expenditure is at least as much as the accuracy improvement}.
\begin{theorem}\label{thm:energyaccuracy}
 $H[S(\omega T)] \geq  h[T] - h[T\mid S(\omega T)]$. 
\end{theorem}
\begin{proof}
The measurement of spin can be viewed as a channel establishing (differential) mutual information $I(T;S(\omega T))$ between $T$ and $S(\omega T)$. Expanding $I(T;S(\omega T))$ two ways, we get:
\begin{align*}
I[T;S(\omega T)]  &=H[S(\omega T)] - H[S(\omega T) \mid T]\
\\&= h[T] - h[T\mid S(\omega T)].
\end{align*}
To conclude the proof, note that $H[S(\omega T) \mid T]\geq 0$.	
\end{proof}

This simple theorem has an interesting physical interpretation. It is well-known in thermodynamics that to reduce entropy in phase space requires work to be done on a system. Theorem~\ref{thm:energyaccuracy} suggests that even to reduce entropy along the time axis, (i.e., when our time-keeping devices are described by time-valued random variables) there may be a similar principle at work. In other words, it suggests that entropy over the time variable also obeys a Szilard-Landauer principle. If such a statement can be proved in much greater generality, it could lead to a thermodynamic theory of clocks. It would also be pleasing from the point of view of Relativity Theory, which requires treating spacetime together rather than separately. 

Taking the thermodynamic analogy further, consider the \textbf{efficiency} 
\[
\eta := (h[T] - h[T\mid S(\omega T)])/H[S(\omega T)].
\]
defined as improvement in accuracy per unit energy cost.

\begin{lemma}\label{lem:eta}
 $\eta$ is a function of $\omega/\lambda$. 
\end{lemma}
\begin{proof}
As in the proof of Theorem~\ref{thm:energyaccuracy}, we can rewrite $\eta = (H[S(\omega T)] - H[S(\omega T) \mid T])/H[S(\omega T)]$. Now $H[S(\omega T)]$ is a function of $\Pr[S(\omega T)=+]$, which is a function of $\omega/\lambda$ and $\phi$ from Equation~\ref{eqn:Sdist}. It remains to show that $H[S(\omega T) \mid T]$ is also a function of $\omega/\lambda$. For every $t\in \mathbb{R}_{\geq 0}$, we have 
\begin{align*}
\Pr[S(\omega T) = + \mid T \in (t,t+dt)] &= \Pr[S(\omega t) = +]\text{ and}\ 
\\\Pr[S(\omega T) = - \mid T \in (t,t+dt)] &= \Pr[S(\omega t) = -] 
\end{align*}
Hence $H[S(\omega T) \mid T]$ is given by

\begin{equation}\label{}
\begin{split}
\int_{t=0}^\infty \bigg(- \Pr[S(\omega t) = +]\log  \Pr[S(\omega t) = +]\\ - \Pr[S(\omega t) = -]\log\Pr[S(\omega t) = -]\bigg)\lambda \e^{-\lambda t}  dt
\end{split}
\end{equation}
Changing the variable of integration to $\kappa = \omega t$, we get
\begin{multline*}
H[S(\omega T) \mid T]= \\
\int_{\kappa=0}^\infty \bigg(- \Pr[S(\kappa) = +]\log \Pr[S(\kappa) = +] -\\   \Pr[S(\kappa) = -]\log  \Pr[S(\kappa) = -] \bigg)\frac{\lambda}{\omega} \e^{-\frac{\lambda}{\omega} \kappa}  d\kappa
\end{multline*}
which is clearly a function of $\omega/\lambda$ since $\kappa$ disappears after integration.
\end{proof}

Because of Lemma~\ref{lem:eta}, we can study efficiency $\eta$ as a function of $\omega/\lambda$ and $\phi$. We obtain a maximum efficiency of $\eta=0.5103$ at $\omega/\lambda = 1.701, \theta=\theta_m=\pi/2$, and $\phi=0.43$. In comparison, if we had operated at the maximum accuracy point by setting $\omega/\lambda = 0.80654$ and $\phi=0.246576$ we would obtain a slightly lower efficiency of $\eta=0.49$.

\section{Nanoscale clock set up}\label{sec:setup}

\begin{figure}
		\includegraphics[width=0.45\textwidth, height=0.4\textwidth]{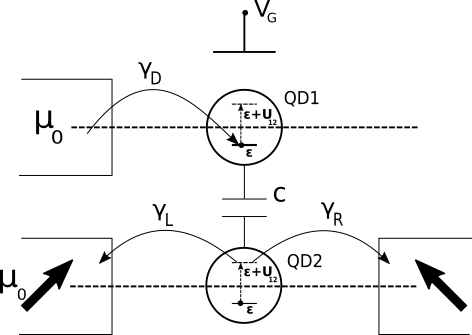}
	\quad
	\caption{Setup for a quantum dot single-qubit clock based on a capacitively coupled quantum dot pair QD1 and QD2. As per our proposal, the exponential prior is the tunneling event in QD1 and the precessing spin is housed inside QD2. Measurement is enabled when the freely precessing spin inside QD2 is "disturbed" by the tunneling event in QD1 via the mutual Coulomb repulsion $U_{12}$. This leads to the spin in QD2 to tunnel into either ferromagnetic contact respectively kept at a chosen quantization axis, leading to the measurement step. \label{fig:sub1}}
	
\end{figure}
We now propose a possible physical realization of the single qubit clock. This is a mesocopic set up comprising of two capacitively coupled quantum dots labeled QD1 and QD2 as shown in Fig.~\ref{fig:sub1}.  A similar set up has been implemented in the context of charge sensing~\cite{Hanson} and single electron memristors \cite{Baugh}. The dot QD1 is coupled weakly to a macroscopic reservoir which we will refer to as the contact. The dot QD2 is coupled to two ferromagnetic contacts whose magnetizations point along the desired measurement axes. We shall now describe how this set up functions as the clock whose goal is to estimate the tunneling time of the electron from the contact to QD1 conditioned on the measurement of a precessing spin housed in QD2.\\
\indent The entire set up at equilibrium is held at a chemical potential $\mu_0$. The tunnel coupling of the dot QD1 to the reservoir is represented via a rate $\gamma_D$ and the tunnel coupling of the dot QD2 to the two ferromagnetic contacts is represented via the rates $\gamma_L$ and $\gamma_R$. The dot QD2 houses the single qubit undergoing stable precession until it is ``disturbed'' by a tunneling event in QD1. This happens due to the long range mutual Coulomb repulsion $U_{12}$ between the electrons in QD1 and QD2.  Due to this, the single particle energies of either dot are $\epsilon$ and $\epsilon+U_{12}$, depending on whether the other dot has an occupied electron. A very large self charging energy of either quantum dot is assumed which prevents further tunneling of electrons from the reservoir, unless the chemical potential $\mu_0$ of the reservoir is raised above the energy that permits double occupation. \\
\indent For the dot QD1, the tunneling times between the contact and the dot are exponentially distributed with a time constant of $\gamma_D^{-1}$. To ensure the sequential nature of the electron tunneling, which is referred to as the sequential tunneling limit in mesoscopic physics \cite{Timm,Hanson,Datta}, the mutual Coulomb interaction energy $U_{12}$, should be much larger than both the coupling energy and the ambient thermal energy, i.e., $U_{12}>>k_BT$ $\&$ $\hbar 
\gamma_D$. An electron tunneling event typically occurs when the dot ground state energy $\epsilon$ is positioned below the chemical potential $\mu$ of the reservoir. This positioning may be tuned via an application of a potential $V_G$ at the gate electrode held close to the dot QD1. We remark that such set ups are very common within current experimental capabilities and are commonly used in spin based quantum computing \cite{Hanson}. \\
\indent The dot QD2 with the two ferromagnetic contacts in the lower half of the schematic serves as the measurement apparatus. Before the tunneling event into QD1 takes place, the electron in QD2 is stable since its ground state energy is lower than the chemical potential of the bottom system. This electronic qubit is made to precess with the application of a magnetic field along an appropriate axis \cite{Basky}. Once the electron tunnels into QD1, the energy level in QD2 is raised by an amount equal to $U_{12}$. If $U_{12}$ is such that the qubit energy level floats above the chemical potential in the bottom half, the precessing electron tunnels into either ferromagnetic contact. Thus we achieve the desired instantaneous coupling between the event, namely the tunneling process in QD1, and the measurement in QD2. Care must be taken that the coupling between QD2 and its reservoirs, $\gamma_L, \gamma_R$ is larger than that of QD1, to ensure the desired sequence of events. A quantum point contact (QPC) detector \cite{Hanson} is stationed near each electrode to ``sense'' whether an electron tunneled to the left or the right contact thereby allowing one to evaluate the necessary probability distributions required to perform the Bayesian inference of the tunneling time. Figure~\ref{fig:sub2} depicts a close up of the physical axes of precession and measurement, with the contacts being oriented along the measurement axes.\\
\section{\label{sec:conclusion}Related work}
Quantum clocks have been previously studied in a pioneering paper by Salecker and Wigner \cite{Wigner}. Their system consists of orthogonal quantum states, one for each digit on a clock face. A unitary evolution takes the system through this sequence of orthogonal quantum states. A projective measurement reports the digit on the clock face as the time. Such clocks were reviewed by Peres in 1979~\cite{peres} where, in addition, he analyzed the perturbative effect of coupling the clock to a physical system. The Salecker-Wigner-Peres clock has found many applications~\cite{leavens1993application,lunardi2011salecker,leavens1994exact,park2009barrier}.

Compared with a two-state version of the Salecker-Wigner clock, instead of merely returning the digit on the face of the clock as the time, we employ Bayesian inference to estimate the posterior distribution, and return its mean as the right estimator for the time. Our approach clarifies the uncertainty involved in timekeeping by explicitly treating timekeeping devices as random variables, and allows analysis of the uncertainty in our estimate of time. We also introduce the idea that it may require energy to keep time. However, we do not consider the perturbative effects that may be introduced when coupling our clock to a physical system to make time measurements. In these aspects, our approach is complementary to the approach of Salecker, Wigner, and Peres.

Our approach towards the study of clocks is influenced by the literature on quantum resource theories and quantum thermodynamics~\cite{horodecki2013quantumness,horodecki2002laws,devetak2008resource,janzing2000thermodynamic,brandao2013resource,brandao2015reversible,brandao2008entanglement,brandao2015second,horodecki2013fundamental,gour6586resource}. One key idea in this literature is to consider thermal equilibrium states as free resources. Analogously, we treat exponential random variables as free resources. Another idea we have borrowed is that of ``one-shot'' processes where thermodynamic questions are examined for single quantum systems instead of for an entire ensemble. The quantum resource theory literature treats questions of reachability and feasibililty. Our work manifests similar ideas in the form of limits on accuracy given certain amounts of resources and energy. Our work can also be viewed in the spirit of Constructor Theory~\cite{deutsch2013constructor,deutsch2015constructor}. 

The work of Rankovic et al.~\cite{rankovic2015quantum} has come to our attention after we prepared this manuscript. They have approached the problem of quantum clocks from a refreshingly fresh direction. They have tackled head-on a question that we have sidestepped: how does one provide an operational meaning to the accuracy of a clock, without relying on any outside, absolute notion of time. They do this by defining a clock's accuracy as the number of alternate ticks two noncommunicating copies of the clock can supply to a third party. They have a notion of an $\epsilon$-continuous quantum clock which appears closely related to our notion of exponential priors. By focusing on a single qubit and on certain simplifying assumptions, we are able to take our analysis to completion and obtain results about limits for a single qubit. In contrast, Rankovic et al. focus on correctly formulating the general problem of timekeeping from an operational point of view. A synthesis of their abstract approach with our concrete one is likely to be of interest.

Two ideas have emerged from the results in this paper: perfectly accurate timekeeping may not be possible with quantum systems, and that reducing uncertainty in time may require energy. Similar ideas have emerged from the work of Erker~\cite{erker2014quantum} through an analysis of quantum hourglasses. Sels and Wouters~\cite{sels2015thermodynamics} have argued that attributing a cost to the measurement of time will establish a second law-like result for unitary dynamics.

\section{Conclusions and future work}
It will be of interest to test our ideas against more general quantum clock constructions. We are tempted to speculate that the following should be a fundamental physical principle: \textbf{Keeping time more accurately than an exponential random variable should require energy proportional to the decrease in entropy from the exponential random variable of the same mean}.

Comparison of our proposal with state-of-the-art metrology standards such as atomic and optical atomic clocks \cite{Essen,Ludlow} seems like another fruitful direction for future work.

Our work is quantum only to the extent that we have used the measurement rule of quantum mechanics. Working with a single qubit allowed us to explore some new ideas with explicit calculations. However, by working only with a single qubit, we have completely ignored entanglement, which is a key feature of quantum mechanics. Many new features of quantum clocks are likely to emerge when one studies larger number of qubits and entanglement.

\begin{widetext}
\appendix
\section{Appendix}\label{sec:app}

In this Appendix, we will derive the expression for the expected quality factor for the General Single Qubit Clock in ~\ref{subsec:q1}. We define the Posterior distributions $T_+$, $T_-$ and $Q[T \mid S(\omega T)]$ analogously to  \ref{subsec:q}.

The posterior distributions are derived using Bayes' Rule~\eqref{eq:bayes}
\begin{equation}
\Pr[T_+ \in (t,t+dt)] = 
\frac{
2 e^{-t }(1 + (\omega/\lambda) ^2) \left|
  e^{i t (\omega/\lambda) /2} \cos \left(\theta_0/2\right) \cos \left(\theta_m/2\right) + 
   e^{i (\phi- t (\omega/\lambda) /2)} \sin  \left(\theta_0/2\right) \sin  \left(\theta_m/2\right)\right|^2}{1 + (\omega/\lambda) ^2 + (1 + (\omega/\lambda) ^2) \cos \theta_0 \cos \theta_m + \cos \phi \sin  \theta_0 \sin  \theta_m + 
 (\omega/\lambda)  \sin  \phi \sin  \theta_0 \sin  \theta_m}
 \end{equation}
 \begin{equation}
 \Pr[T_- \in (t,t+dt)] = \frac{2 e^{-t} (1 + (\omega/\lambda)^2)\left(1- \left|
  e^{it (\omega/\lambda)/2} \cos   \left(\theta_0/2\right) \cos   \left(\theta_m/2\right) + 
   e^{i(\phi- (t (\omega/\lambda))/2)} \sin  \left(\theta_0/2\right) \sin  \left(\theta_m/2\right)\right|^2\right)}{
   1 + (\omega/\lambda)^2 - (1 + (\omega/\lambda)^2) \cos   \theta_0 \cos   \theta_m - \cos   \phi \sin  \theta_0 \sin  \theta_m - 
 (\omega/\lambda) \sin  \phi \sin  \theta_0 \sin  \theta_m}
  \end{equation}
  The mean of the posterior distributions are then

\begin{equation}
\begin{split}
  E[T_+] &=    \int_{t=0}^{\infty} t \Pr[T_{+} \in (t, t+dt)]dt \\
  &=  \frac{1}{\lambda}\left[\frac {(1 +  (\omega/\lambda)^2)^2(1 + \cos \theta_0  \cos \theta_m  )  + (\sin \theta_0  \sin \theta_m)(\cos \phi    - 
   (\omega/\lambda)^2 \cos \phi    +2 (\omega/\lambda) \sin \phi  )}{(1 + (\omega/\lambda)^2) ((1 + (\omega/\lambda)^2)( 1  + 
     \cos \theta_0  \cos \theta_m ) + (  \sin \theta_0  \sin \theta_m )(\cos \phi +
     (\omega/\lambda) \sin \phi  ))}\right]
\end{split}
\end{equation}
\begin{equation}
\begin{split}
  E[T_-]  &=    \int_{t=0}^{\infty} t \Pr[T_{-} \in (t, t+dt)]dt \\
  &=  \frac{1}{\lambda}\left[\frac {(1 +  (\omega/\lambda)^2)^2(1 - \cos \theta_0  \cos \theta_m  )  - (\sin \theta_0  \sin \theta_m)(\cos \phi    -
   (\omega/\lambda)^2 \cos \phi    +2 (\omega/\lambda) \sin \phi  )}{(1 + (\omega/\lambda)^2) ((1 + (\omega/\lambda)^2)( 1  - 
     \cos \theta_0  \cos \theta_m ) - (  \sin \theta_0  \sin \theta_m )(\cos \phi +
     (\omega/\lambda) \sin \phi  ))}\right]
\end{split}
\end{equation}
\begin{equation}
\begin{split}
  E[T_+^2] =  \int_{t=0}^{\infty} t^2 \Pr[T_{+} \in (t, t+dt)]dt=\frac{1}{\lambda^2}\left[\dfrac {\splitdfrac{2 ((1 +  (\omega/\lambda)^2 )^3(1+ \cos  \theta_0 \cos  \theta_m)+ (\sin \theta_0 \sin \theta_m)(\cos  \phi  - }{ 3 (\omega/\lambda)^2 \cos  \phi + 3 (\omega/\lambda) \sin \phi - 
     (\omega/\lambda)^3 \sin \phi))}}{\splitdfrac{(1 + (\omega/\lambda)^2)^2 ((1 + (\omega/\lambda)^2)(1 + 
     \cos  \theta_0 \cos  \theta_m) +}{(\sin \theta_0 \sin \theta_m)( \cos  \phi  + 
     (\omega/\lambda) \sin \phi ))}}\right]
\end{split}
\end{equation}
\begin{equation}
\begin{split}
  E[T_-^2] =   \int_{t=0}^{\infty} t^2 \Pr[T_{-} \in (t, t+dt)]dt=\frac{1}{\lambda^2}\left[\dfrac {\splitdfrac{2 ((1 +  (\omega/\lambda)^2 )^3(1- \cos  \theta_0 \cos  \theta_m)- (\sin \theta_0 \sin \theta_m)(\cos  \phi  - }{ 3 (\omega/\lambda)^2 \cos  \phi + 3 (\omega/\lambda) \sin \phi - 
     (\omega/\lambda)^3 \sin \phi))}}{\splitdfrac{(1 + (\omega/\lambda)^2)^2 ((1 + (\omega/\lambda)^2)(1 - 
     \cos  \theta_0 \cos  \theta_m) -}{(\sin \theta_0 \sin \theta_m)( \cos  \phi  + 
     (\omega/\lambda) \sin \phi ))}}\right]
\end{split}
\end{equation}

The Expected Quality Factor can then be computed from ~\ref{subsec:rv} and~\eqref{eq:avgq}. Maximizing the expected quality factor as a function of $\theta_0$ and $\theta_m$ leads to the result $\theta_0=\theta_m=\pi/2$. In this case, the expected quality factor is: 
\begin{multline}
\frac{1}{\sqrt{2} (1 + \gamma^2)^2} \left[\frac{(1 + \gamma^2)^2 + (-1 + \gamma^2) \cos(\phi) - 2 \gamma \sin(\phi)}{
  \sqrt{\frac{
  3 + 6 \gamma^2 + 9 \gamma^4 + 8 \gamma^6 + 2 \gamma^8 - 
   4 (1 + \gamma^4 + 2 \gamma^6) \cos(\phi) + (1 - 6 \gamma^2 + \gamma^4) \cos(2 \phi) -4( 
   2  + \gamma^2  + \gamma^6 )\gamma\sin(\phi) + 4( 1  - 
    \gamma^2 )\gamma\sin(2 \phi)}{(1 + \gamma^2)^2 (-1 - \gamma^2 + \cos(\phi) + 
     \gamma \sin(\phi))^2}}} \right.\\
   \left.  + \frac{(1 + \gamma^2)^2 - (-1 + \gamma^2) \cos(\phi) + 2 \gamma \sin(\phi)}{
  \sqrt{\frac{
  3 + 6 \gamma^2 + 9 \gamma^4 + 8 \gamma^6 + 2 \gamma^8 + 
   4 (1 + \gamma^4 + 2 \gamma^6) \cos(\phi) + (1 - 6 \gamma^2 + \gamma^4) \cos(2 \phi) + 
   4(2 +  \gamma^2  +  \gamma^6 )\gamma\sin(\phi) + 4( 1  - 
    \gamma^2 )\gamma\sin(2 \phi)}{(1 + \gamma^2)^2 (1 + \gamma^2 + \cos(\phi) + \gamma \sin(\phi))^2}}}\right]
\end{multline}
where $\gamma=\omega/\lambda$. Similarly, maximizing the efficiency as a function of $\theta_0$ and $\theta_m$ leads to $\theta_0=\theta_m=\pi/2$.
\end{widetext}

\bibliographystyle{IEEEtran}
\bibliography{ref_clock}

\end{document}